\documentclass[12pt,american,a4paper, 12pt]{amsart}
\usepackage[T1]{fontenc}
\usepackage[latin9]{inputenc}
\usepackage{verbatim}
\usepackage{amstext}
\usepackage{amsthm}
\usepackage{amssymb}
\usepackage{graphicx}

\makeatletter
\numberwithin{equation}{section}
\numberwithin{figure}{section}
\theoremstyle{plain}
\newtheorem{thm}{\protect\theoremname}
\theoremstyle{definition}
\newtheorem{defn}[thm]{\protect\definitionname}
\theoremstyle{plain}
\newtheorem{prop}[thm]{\protect\propositionname}
\theoremstyle{plain}
\newtheorem{cor}[thm]{\protect\corollaryname}
\theoremstyle{definition}
\newtheorem{example}[thm]{\protect\examplename}
\theoremstyle{remark}
\newtheorem{rem}[thm]{\protect\remarkname}

\@ifundefined{date}{}{\date{}}
\makeatother

\usepackage{babel}
\providecommand{\corollaryname}{Corollary}
\providecommand{\definitionname}{Definition}
\providecommand{\examplename}{Example}
\providecommand{\propositionname}{Proposition}
\providecommand{\remarkname}{Remark}
\providecommand{\theoremname}{Theorem}

\begin{document}

\title{Riemann-Hilbert factorization of matrices invariant under inversion
in a circle}

\author{Hideshi Yamane}

\curraddr{{\small{}Department of Mathematical Sciences, Kwansei Gakuin University
}\\
{\small{}Gakuen 2-1 Sanda, Hyogo 669-1337, Japan}}

\email{{\small{}yamane@kwansei.ac.jp}}
\begin{abstract}
We consider matrix functions with certain invariance under inversion
in the unit circle. If such a function satisfies a positivity assumption
on the unit circle, then only zero partial indices appear in its Riemann-Hilbert
(Wiener-Hopf) factorization. It implies the unique solvability of
a certain class of Riemann-Hilbert boundary value problems. It includes
the ones associated with the inverse scattering transform of the focusing/defocusing
integrable discrete nonlinear Schr\"odinger equations. \\
AMS subject classfication: 35Q15, 47A68 
\end{abstract}


\maketitle
\markboth{Hideshi Yamane}{Riemann-Hilbert factorization}

\section{Introduction}

Riemann-Hilbert problems (RHPs), formulated in various ways, are a
powerful tool in the study of integrable systems. As is proved in
\cite{ZhouSIAM89}, if a matrix function is invariant under Schwarz
reflection, its Riemann-Hilbert (Wiener-Hopf) factorization involves
only zero partial indices and it implies the unique solvability of
the corresponding RHPs formulated in other ways (a singular integral
equation and a boundary value problem). The zero partial indices property
is a key in the argument in \cite{CPAM89}. The main result there
is the bijectivity of the scattering and inverse scattering maps.
So bijectivity is known for NLS in sufficient detail (see also \cite{DZ2003,CPAM98}),
but the integrable discrete nonlinear Schr\"odinger equation (IDNLS)
still lacks a satisfactory theory. In order to construct such a theory,
we need a detailed information about relevant RHPs. In the discrete
case, the real axis must be replaced by the unit circle (\cite{APT,IDNLSdefocusing,IDNLSdefocusing2,IDNLS}).
If a matrix function invariant under the inversion in $S^{1}$, namely
$z\to1/\bar{z},$ and satisfies a positivity condition on $S^{1}$,
then it has only zero partial indices. It implies the unique solvability
of a certain class of Riemann-Hilbert boundary value problems including
the ones associated with IDNLS. This fact can be a basis of the bijectivity
proof of the scattering/inverse scattering transforms for IDNLS. See
also the approach in \cite{FokasXia} based on a vanishing lemma. 

Factorization of matrices given on $S^{1}$ is a topic that can be
studied from other directions. See, e.g., \cite{clanceygohberg,Lit}.
It is known that a positive Hermitian matrix function $v$ on $S^{1}$
has an expression $v=w^{*}w$, where $w$ is holomorphic inside $S^{1}$.
We give a generalization of this fact to the case of an inversion
invariant contour including $S^{1}$. 

\section{Function spaces}

Let $\Sigma\subset\mathbb{C}$ be a finite disjoint union of smooth
simple closed curves. More specifically, we assume $\Sigma=\cup_{j=1}^{J}\Sigma_{j}$,
where each $\Sigma_{j}$ is a smooth simple closed curve and $\Sigma_{j}\cap\Sigma_{k}=\emptyset\,(j\ne k)$.
It is possible to assign an orientation on $\Sigma$ such that it
is the positively oriented boundary of an open set $\Omega_{+}$.
Set $\Omega_{-}=\mathbb{C}\setminus(\Sigma\cup\Omega_{+})$. Then
$\Sigma$ is the negatively oriented boundary of the open set $\Omega_{-}.$
We introduce function spaces following \cite{ZhouSIAM89,CPAM89}.
The $L^{2}$ norm of a matrix function $f\colon\Sigma\to\mathbf{M}_{n}$($\mathbf{M}_{n}$
is the complex $n\times n$ matrix algebra) is defined by $\|f\|_{2}=\left(\int_{\Sigma}|f|^{2}|dz|\right)^{1/2},|f|=(\mathrm{tr}f^{*}f)^{1/2}$,
where the asterisk means the Hermitian conjugate. We write $L^{2}(\Sigma)$
for $L^{2}(\Sigma,\mathbf{M}_{n})$. We denote by $H^{k}(\Sigma)\,(k\ge1)$
the space of all the matrix functions $f$ such that $f^{(j)}\in L^{2}(\Sigma)$
for all $j=0,\dots,k$ in the distribution sense. Its norm is $\|f\|_{2,k}=\left(\sum_{j=0}^{k}\|f^{(j)}\|_{2}^{2}\right)^{1/2}$
and $H^{k}(\Sigma)$ is a Hilbert space with continuous pointwise
multiplication. Sometimes we write $\|f\|_{2,k}$ as $\|f\|_{k}$
for brevity. A function $f\in H^{k}(\Sigma)$ is H\"older continuous. 

In the present paper, we choose a formulation in which the contour
is bounded. In \cite{ZhouSIAM89}%
{} and \cite{CPAM89}%
, however, the author assumes that the contour is unbounded. At some
places of the present paper, we reduce the proof to the unbounded
case. %
{} Let $C_{\pm}$ be the Cauchy operators defined by 
\[
C_{\pm}f(z)=\lim_{z'\to z}\frac{1}{2\pi}\int_{\Sigma}\frac{f(w)\,dw}{w-z'},
\]
where the nontangential limit $z'\to z$ is taken from $\Omega_{\pm}$
respectively. They are bounded from $L^{2}(\Sigma)$ to $L^{2}(\Sigma)$
and from $H^{k}(\Sigma)$ to $H^{k}(\Sigma)$. It is known that $C_{+}$
and $-C_{-}$ are complementary projections. Moreover, a function
in $\mathrm{Ker}\,C_{\pm}=\mathrm{Range}\,C_{\mp}$ has a holomorphic
extention to $\Omega_{\mp}.$

\section{Formulation of Riemann-Hilbert problems}

Assume that $v\in H^{k}(\Sigma)$ admits a factorization $v=(b^{-})^{-1}b^{+}$
for invertible $b^{\pm}\in H^{k}(\Sigma)$. Since  $\Sigma$ is bounded,
$|\det b^{-}|$ is uniformly away from 0 and $(b^{-})^{-1}\in H^{k}(\Sigma)$.
A factorization as above always exists (we have only to choose $b^{+}=I$
or $b^{-}=I$). Set $w^{\pm}=\pm(b^{\pm}-I)$, i.e. $b^{\pm}=I\pm w^{\pm}$.
We call $w=(w^{+},w^{-})$ \textsl{a pair of factorization data} of
$v$. The set of all such pairs is denoted by $FD_{k}$. We have 
\[
FD_{k}=\left\{ \left(w^{+},w^{-}\right)\in\oplus^{2}H^{k}(\Sigma);\,I\pm w^{\pm}\text{ is invertible}\right\} .
\]

\begin{defn}
An element $\mu\in H^{j}(\Sigma)(j=0,\dots,k)$ is said to be a solution
of the Riemann-Hilbert problem (RHP) of the pair of factorization
data $w$ if 
\begin{equation}
\mu b^{\pm}-h\in\mathrm{Range}\,C_{\pm}\label{eq:RHPdef}
\end{equation}
for some constant matrix $h$. 
\end{defn}

The definition in \cite{ZhouSIAM89} has been modified here because
$\mu(\infty)$ is not defined in the present paper. Notice that $m_{\pm}:=\mu b^{\pm}\in H^{j}(\Sigma)$.
Since  they are in the ranges of the Cauchy operators modulo $h$,
they have a holomophic extension to $\mathbf{C}\setminus\Sigma$,
which we denote by $m$. We call it the solution of the Riemann-Hilbert
problem of $v$ or $w$. 
\begin{prop}
\label{prop: RHPclassical}If $\mu$ is a solution of (\ref{eq:RHPdef})
for fixed $h$, then the holomophic extension $m$ is a solution of
a Riemann-Hilbert boundary value problem in the classical sense: 
\[
m_{+}=m_{-}v,m(\infty)=\lim_{z\to\infty}m(z)=h.
\]
Conversely, if a holomorphic function $m$ satisfies $m(\infty)=h$
and $m_{+}=m_{-}v$, then $\mu=m_{+}(b^{+})^{-1}=m_{-}(b^{-})^{-1}$
is a solution of (\ref{eq:RHPdef}). 
\end{prop}

\begin{proof}
We have $m_{+}(b^{+})^{-1}=\mu=m_{-}(b^{-})^{-1}$ and $m_{+}=m_{-}(b^{-})^{-1}b^{+}=m_{-}v$.
Next, $m(\infty)=h$ follows from $m=h+\text{(a Cauchy integral)}$.
The converse is now easy. 
\end{proof}
For $w=(w^{+},w^{-})$, set 
\[
C_{w}\phi=C_{+}(\phi w^{-})+C_{-}(\phi w^{+}).
\]
Then $C_{w}$ is a bounded operator from $H^{j}(\Sigma)$ to itself
for every $j=0,1,\dots,k$. 
\begin{prop}
\label{prop:RHP and resolvent}An element $\mu$ of $L^{2}(\Sigma)$
is a solution of (\ref{eq:RHPdef}) if and only if 

\begin{equation}
(I-C_{w})\mu=h\label{eq: singintegraleq}
\end{equation}
 holds. If $\mathrm{Id}-C_{w}$ is a bijection, then a solution of
(\ref{eq:RHPdef}) exists uniquely. 
\end{prop}

\begin{proof}
We follow the proof of \cite[Prop 3.3]{ZhouSIAM89}. Recall that $C_{+}-C_{-}=\mathrm{Id}.$
If $\mu$ satisfies (\ref{eq: singintegraleq}), we have
\begin{align*}
\mu b^{+}-h & =\mu(I+w^{+})-(I-C_{w})\mu=\mu w^{+}+C_{w}\mu\\
 & =(C_{+}-C_{-})(\mu w^{+})+C_{+}(\mu w^{-})+C_{-}(\mu w^{+})\\
 & =C_{+}(\mu w^{+}+\mu w^{-})\in\mathrm{Range}\,C_{+,}
\end{align*}
and similarly $\mu b^{-}-h=C_{-}(\mu w^{+}+\mu w^{-})\in\mathrm{Range}\,C_{-}$.
\\
Conversely, assume (\ref{eq:RHPdef}). Then $\mu b^{\pm}-h\in\mathrm{Ker}\,C_{\mp}.$
We have
\begin{align*}
(I-C_{w})\mu & =(C_{+}-C_{-})\mu-\left[C_{+}(\mu w^{-})+C_{-}(\mu w^{+})\right]\\
 & =C_{+}(\mu b^{-})-C_{-}(\mu b^{+})\\
 & =C_{+}(\mu b^{-}-h)-C_{-}(\mu b^{+}-h)+h=h.
\end{align*}
\end{proof}

\section{Factorization and partial indices}

We introduce two classes of holomorphic matrix functions following
\cite{CPAM89}: 
\begin{eqnarray*}
\mathcal{H}^{k}(\mathbf{C}\setminus\Sigma) & := & \left\{ m;\,m_{\pm}-m(\infty)\in\mathrm{Range}\,C_{\pm}\right\} ,\\
G\mathcal{H}^{k}(\mathbf{C}\setminus\Sigma) & := & \left\{ m\in\mathcal{H}^{k}(\mathbf{C}\setminus\Sigma);\,\det m\,\mathrm{vanishes\,nowhere}\right\} ,
\end{eqnarray*}
where $C_{\pm}\colon H^{k}(\Sigma)\to H^{k}(\Sigma)$. 
\begin{thm}
\label{thm:factorization}Any $v\in H^{k}(\Sigma)$ with $\det v\ne0$
admits a Riemann-Hilbert (Wiener-Hopf) factorization $v=m_{-}^{-1}\theta m_{+}$
relative to $\Sigma$ in $H^{k}(\Sigma)$. Here $m_{\pm}$ are the
boundary values of an element $m$ of $G\mathcal{H}^{k}(\mathbf{C}\setminus\Sigma)$.
The matrix $\theta$ is 
\begin{equation}
\theta=\mathrm{diag}\left[\left(\frac{z-z_{+}}{z-z_{-}}\right)^{k_{1}},\dots,\left(\frac{z-z_{+}}{z-z_{-}}\right)^{k_{n}}\right],\label{eq:theta}
\end{equation}
where $z_{\pm}\in\Omega^{\pm}$ and $k_{1},\dots,k_{n}$ are integers
such that $k_{1}\ge\dots\ge k_{n}$. We call $k_{1},\dots,k_{n}$
the partial indices of $v$. They are uniquely determined. 
\end{thm}

\begin{proof}
Fix $z_{\pm}\in\Omega^{\pm}$. We embed our contour $\Sigma$ into
$\widehat{\Sigma}\ni\infty$ and reduce the proof to \cite[Th 9.1]{ZhouSIAM89}
or \cite[Th 2.1.3]{CPAM89}. Let $\Sigma'$ be a line (a circle in
the Riemann sphere) defined by $\mathrm{Re\,}z=-p$, where $p$ is
so large that $\Sigma'$ is far away from $\Sigma$ and $z_{\pm}$.
First we assume that $\Omega_{+}$ is bounded and that $\Omega_{-}$
is unbounded. The orientation of $\Sigma'$ is from $-p-i\infty$
to $-p+i\infty$. If $p$ is sufficiently large, we have $z_{+}\in\Omega_{+}\subset\widehat{\Omega}{}_{+},z_{-}\in\widehat{\Omega}_{-}\subset\Omega_{-}$.
Set $\widehat{\Sigma}=\Sigma\cup\Sigma'$. It has a compatible orientation
in the sense that it is a positively oriented boundary of an open
set $\widehat{\Omega}{}_{+}$ and is a negatively oriented boundary
of an open set $\widehat{\Omega}_{-}$. Extend $v\in H^{k}(\Sigma)$
to $\widehat{\Sigma}$ by setting $\hat{v}|_{\Sigma}=v,\hat{v}|_{\Sigma'}=I$.
 Then $\hat{v}$ is not an element of $H^{k}(\widehat{\Sigma})$,
but it belongs to $H_{I}^{k}(\widehat{\Sigma})=H^{k}(\Sigma)\oplus\mathbf{M}_{n}$
introduced in \cite{ZhouSIAM89} and \cite{CPAM89}. It consists of
matrix functions $f$ on $\widehat{\Sigma}$ with the limit $f(\infty)$
such that $f-f(\infty)\in H^{k}(\widehat{\Sigma})$. The norm is the
square root of $|f(\infty)|^{2}+\|f-f(\infty)\|_{k}^{2}$. Since  there
is no self-intersection, it is not necessary to introduce $H^{k}(\widehat{\Sigma}{}^{\pm})$
and $H^{k}(\Sigma^{\pm})$. 

By \cite[Th 9.1]{ZhouSIAM89} or \cite[Th 2.1.3]{CPAM89}, $\hat{v}\in H_{I}^{k}(\widehat{\Sigma})$
admits a Riemann-Hilbert factorization 
\begin{align*}
\hat{v} & =\hat{m}_{-}^{-1}\theta\hat{m}_{+},\\
\theta & =\mathrm{diag}\left[\left(\frac{z-z_{+}}{z-z_{-}}\right)^{k_{1}},\dots,\left(\frac{z-z_{+}}{z-z_{-}}\right)^{k_{n}}\right].
\end{align*}

\includegraphics[scale=0.5]{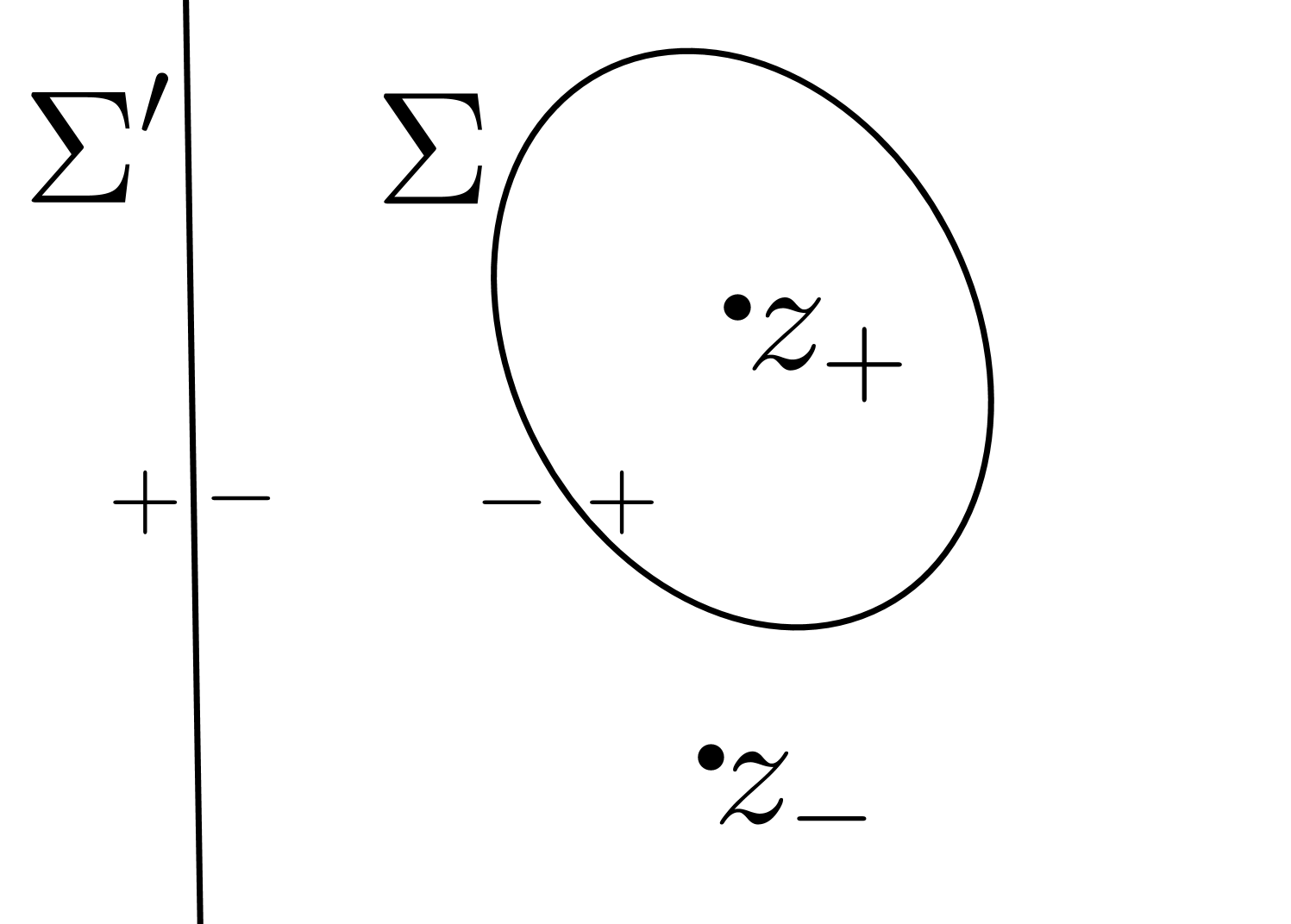}

On $\Sigma'$, we have $\hat{v}=I=\hat{m}_{-}^{-1}\theta\hat{m}_{+}$,
which implies $\hat{m}_{+}=\theta^{-1}\hat{m}_{-}$. Set $m=\theta^{-1}\hat{m}$
in $\mathrm{Re\,}z<-p$ (the positive side of $\Sigma'$) and $m=\hat{m}$
elsewhere. Then $m$ is holomorphic for $z\not\in\Sigma$ and we have
a factorization $v=m_{-}^{-1}\theta m_{+}$ on $\Sigma$. In particular,
$v$ and $\hat{v}$ has the same partial indices. \\
Next, if $\Omega_{+}$ is unbounded and $\Omega_{-}$ is bounded,
we reverse the orientation of $\Sigma'$. We get $\hat{m}_{-}=\theta\hat{m}_{+}$
on $\Sigma'$ and set $m=\theta\hat{m}$ in $\mathrm{Re\,}z<-p$ . 
\end{proof}
\begin{thm}
\label{thm: KerCoker}The operator $\mathrm{Id}-C_{w}\colon H^{k}(\Sigma)\to H^{k}(\Sigma)$
is Fredholm. Let $k_{1},\dots,k_{n}$ be the partial indices of $v$.
Then 
\begin{align*}
\mathrm{\dim\mathrm{Ker}(\mathrm{Id}}-C_{w}) & =n\sum_{k_{j}>0}k_{j},\\
\mathrm{\dim\mathrm{Coker}(\mathrm{Id}}-C_{w}) & =-n\sum_{k_{j}<0}k_{j}.
\end{align*}
\end{thm}

\begin{proof}
We employ the embedding argument in the proof of Theorem \ref{thm:factorization}.
We extend $v\in H^{k}(\Sigma)$ and $w$ to $\widehat{\Sigma}$ by
setting $\hat{v}|_{\Sigma}=v,\hat{v}|_{\Sigma'}=I$ and $\hat{w}|_{\Sigma}=w,\hat{w}|_{\Sigma'}=(0,0)$.
Then $\hat{w}$ is a pair of factorization data of $\hat{v}$. Recall
that $v$ and $\hat{v}$ has the same partial indices. 

On $H_{I}^{k}(\widehat{\Sigma})=H^{k}(\Sigma)\oplus H_{I}^{k}(\Sigma')$,
we have $C_{\hat{w}}=C_{w}\oplus0$. By \cite[Th 9.2]{ZhouSIAM89}
and \cite[Th 2.1.6]{CPAM89}, we have $\mathrm{\dim\mathrm{Ker}(\mathrm{Id}}-C_{w})\mathrm{=\dim\mathrm{Ker}(\mathrm{Id}}-C_{\hat{w}})=n\sum_{k_{j}>0}k_{j}$.
The assertion about the cokernel is proved in the same way.
\end{proof}
\begin{cor}
\label{cor: RHPuniquesol}If the partial indices are all zero, the
Riemann-Hilbert problem (\ref{eq:RHPdef}) has a unique solution. 
\end{cor}

\begin{proof}
Use Proposition \ref{prop:RHP and resolvent} and Theorem \ref{thm: KerCoker}. 
\end{proof}

\section{Inversion in the unit circle}

For a subset $A$ of $\mathbf{C}$ and a matrix function $f$, we
set $A^{\sharp}=\left\{ 1/\bar{z};\,z\in A\right\} $ and $f^{\sharp}(z)=f(1/\bar{z})^{*}$.
It is the inversion in the unit circle $S^{1}=\left\{ z;\,|z|=1\right\} $.
For example, if $\theta$ is as in (\ref{eq:theta}) and $z_{\pm}\ne0$,
then we have 
\begin{equation}
\theta^{\sharp}=\mathrm{diag}\left[\left(\frac{\bar{z}_{+}}{\bar{z}_{-}}\cdot\frac{z-1/\bar{z}_{+}}{z-1/\bar{z}_{-}}\right)^{k_{1}},\dots,\left(\frac{\bar{z}_{+}}{\bar{z}_{-}}\cdot\frac{z-1/\bar{z}_{+}}{z-1/\bar{z}_{-}}\right)^{k_{n}}\right].\label{eq:thetainversion}
\end{equation}

\begin{thm}
\label{thm: inversion}Let $\Sigma\supset S^{1}$ be a contour invariant
under inversion in $S^{1}$. If $v\in H^{k}(\Sigma)$ with $\det v\ne0$
satisfies 
\[
v=v^{\sharp}\;\text{on}\;\Sigma\setminus S^{1}\;\text{ and }\;\mathrm{Re}\,v=(v+v^{*})/2>0\;\text{on}\;S^{1},
\]
then it has only zero partial indices and a solution of (\ref{eq:RHPdef})
exists uniquely when $h$ is fixed.
\end{thm}

\begin{proof}
Let $w=(w^{+},w^{-})$ be an arbitrary pair of factorization data
of $v$ and set $b^{\pm}=I\pm w^{\pm}$. It is enough to prove the
bijectivity of $\mathrm{Id}-C_{w}$. Let $v=m_{-}^{-1}\theta m_{+}$
be the factorization of $v$ on $\Sigma$ as in Theorem \ref{thm:factorization}.
By inversion, we have $v^{\sharp}=m_{+}^{\sharp}\theta^{\sharp}(m_{-}^{-1})^{\sharp}$
on $\Sigma^{\sharp}=\Sigma$. Since $1/\bar{z}_{\pm}\in\Omega^{\mp}$,
$m_{+}^{\sharp}\in\mathrm{Range\,}C_{-}$, $(m_{-}^{-1})^{\sharp}\in\mathrm{Range\,}C_{+}$,
the expression (\ref{eq:thetainversion}) of $\theta^{\sharp}$ implies
that the partial indices of $v^{\sharp}$ are $-k_{n},\dots,-k_{1}$.
If $w^{\sharp}$ is a pair of factorizationdata of $v^{\sharp}$,
we have by Theorem \ref{thm: KerCoker}
\[
\mathrm{\dim\mathrm{Coker}(\mathrm{Id}}-C_{w})=\mathrm{\dim\mathrm{Ker}(\mathrm{Id}}-C_{w^{\sharp}}).
\]
Since  $v^{\sharp}$ also  satisfies the assumptions of the theorem,
it is enough to prove that $\mathrm{Ker(\mathrm{Id}}-C_{w})=0$. 

\includegraphics[scale=0.6]{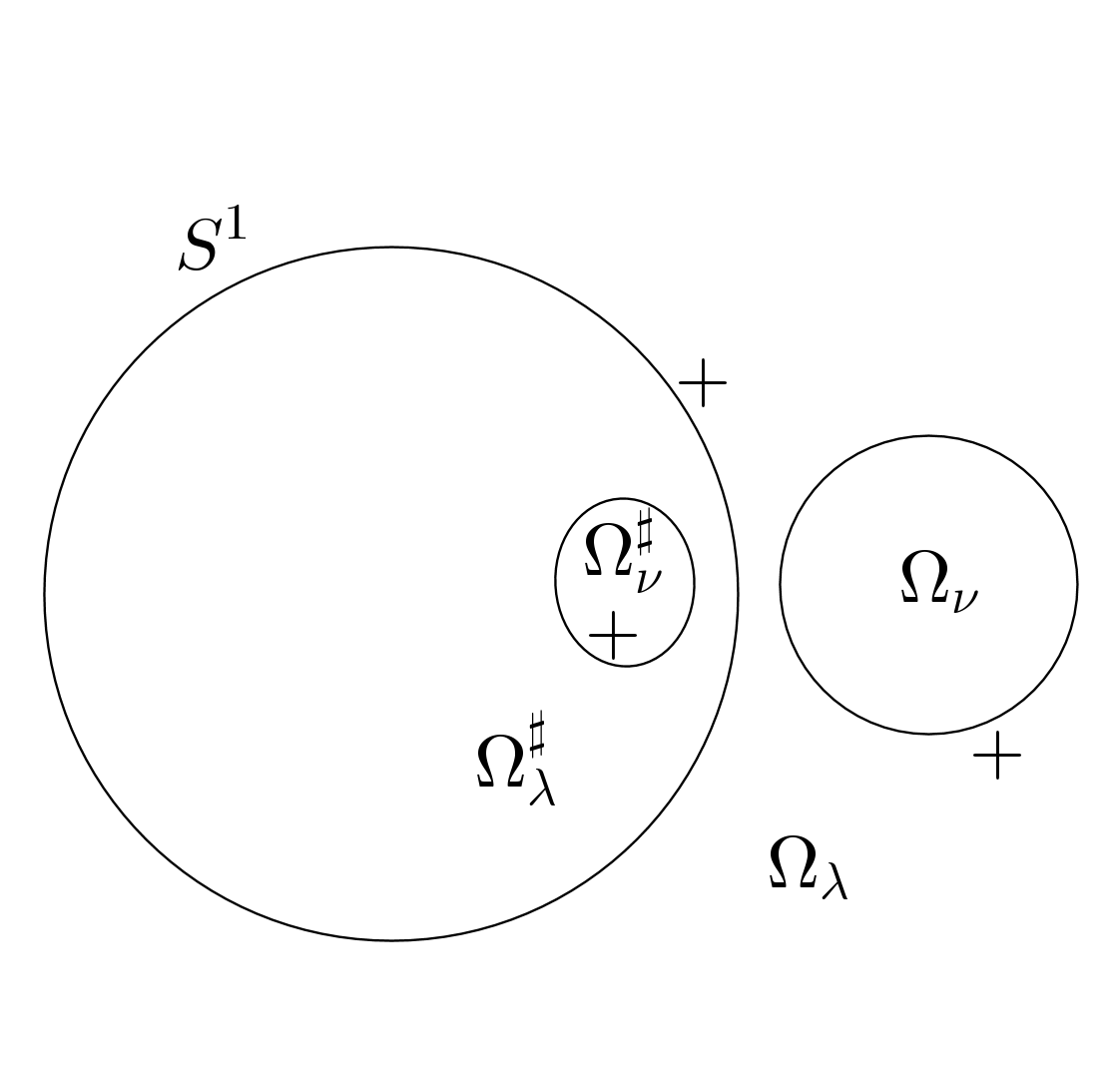}

Let $\Omega_{\nu}$ be a component of $\mathbf{C}\setminus\Sigma$
outside $S^{1}$. In the figure, the orientation of $\Sigma$ is indicated
by placing plus signs on the positive sides of the curves. We may
assume that $S^{1}$ has the clockwise orientation following the convention
of \cite{APT}. Assume $\mu\in\mathrm{Ker}(\mathrm{Id}-C_{w})$. Then
by Proposition \ref{prop:RHP and resolvent}, we have $m_{\pm}:=\mu b^{\pm}\in\mathrm{Range}\,C_{\pm}$
and they have a holomorphic extension, which we denote by $m$. Let
$m_{\nu1},m_{\nu2}$ be the boundary values of $m|_{\Omega_{\nu}},m|_{\Omega_{\nu}^{\sharp}}$
respectively. If $\Omega_{\nu}\subset\Omega_{\pm}$, then $\Omega_{\nu}^{\sharp}\subset\Omega_{\mp}$.
So if $m_{\nu1}$is the boundary value from (a part of) $\Omega_{\pm}$,
then $m_{\nu2}$ is the boundary value from (a part of) $\Omega_{\mp}$.
We have 
\[
\int_{\partial\Omega_{\nu}}m_{\nu1}m_{\nu2}^{\sharp}=0.
\]
Here notice that the usual counterclockwise orientation of $\partial\Omega_{\nu}$
may or may not coincide with the orientation of $\Sigma$ depending
on whether $\Omega_{\nu}\subset\Omega_{+}$ or $\Omega_{\nu}\subset\Omega_{-}$.
We calculate the sum with respect to all the components $\Omega_{\nu}$
outside $S^{1}$, including the one whose boundary contains $S^{1}$
(like $\Omega_{\lambda}$ in the figure). We get 
\begin{equation}
\sum_{\nu}\int_{\partial\Omega_{\nu}}m_{\nu1}m_{\nu2}^{\sharp}=0.\label{eq:sum mmintegral}
\end{equation}
It is possible that a single curve, not $S^{1}$, is included in both
$\partial\Omega_{\nu}^{+}$ with $\Omega_{\nu}^{+}\subset\Omega^{+}$
and $\partial\Omega_{\nu}^{-}$ with $\Omega_{\nu}^{-}\subset\Omega^{-}$
. In this case it has two orientations. So cancellation happens in
the sum above. Now we show 
\begin{equation}
\sum_{\nu}\int_{\partial\Omega_{\nu}}m_{\nu1}m_{\nu2}^{\sharp}=\int_{S^{1}}m_{-}vm_{-}^{\sharp}.\label{eq:sum mmintegral-1}
\end{equation}
Let $\Sigma_{\nu}$ be a component of $\Sigma$ outside $S^{1}$.
Then it is a part of the common boundary of some components $\Omega_{\nu}^{+}\subset\Omega^{+}$
and $\Omega_{\nu}^{-}\subset\Omega^{-}$. Let $m_{1}^{+},m_{1}^{-}$
be the boundary values of $m$ on $\Sigma_{\nu}$ from $\Omega_{\nu}^{+},\Omega_{\nu}^{-}$
respectively (hence $m_{1}^{+}=m_{1}^{-}v$) and let $m_{2}^{+},m_{2}^{-}$
be the boundary values of $m$ on $\Sigma_{\nu}^{\sharp}$ from $\Omega_{\nu}^{-\sharp},\Omega_{\nu}^{+\sharp}$
respectively (hence $m_{2}^{+}=m_{2}^{-}v$). In the left-hand side
of (\ref{eq:sum mmintegral-1}), the integral along $S^{1}$ appear
only once as $\int_{S^{1}}m_{+}m_{-}^{\sharp}=\int_{S^{1}}m_{-}vm_{-}^{\sharp}$.
The integrals along $\Sigma_{\nu}$ appear twice, once as $\int_{\Sigma_{\nu}}m_{1+}m_{2-}^{\sharp}=\int_{\Sigma_{\nu}}m_{1-}vm_{2-}^{\sharp}$
and once again as $\int_{-\Sigma_{\nu}}m_{1-}m_{2+}^{\sharp}=\int_{-\Sigma_{\nu}}m_{1-}v^{\sharp}m_{2-}^{\sharp}$.
Since  $v=v^{\sharp}$ on $\Sigma_{\nu}$, these integrals cancel
each other and (\ref{eq:sum mmintegral-1}) has been proved. By (\ref{eq:sum mmintegral})
and (\ref{eq:sum mmintegral-1}), we have 
\[
\int_{S^{1}}m_{-}vm_{-}^{*}=\int_{S^{1}}m_{-}vm_{-}^{\sharp}=0.
\]
Inversion (or Hermitian conjugation) gives $\int_{S^{1}}m_{-}v^{*}m_{-}^{*}=0$.
Adding these two equations, we get 
\[
\int_{S^{1}}m_{-}(v+v^{*})m_{-}^{*}=0.
\]
By the positivity of $\mathrm{Re}\,v$, we have $m_{-}=0$ on $S^{1}$,
which implies $m_{+}=m_{-}v=0$ there. We get $m=0$ at least in the
components of $\Omega_{\pm}$ whose boundaries include $S^{1}$ like
$\Omega_{\lambda}$ and $\Omega_{\lambda}^{\sharp}$ in the figure.
Then the boundary value $m_{+}$ or $m_{-}$ from such a component
vanishes along any other part of the boundary. Since  $v$ is invertible,
the boundary value from the other side  also  vanishes and we have
$m=0$ in that side. We can repeat this process as many times as necessary
(e.g. concentric circles) and finally we get $m_{\pm}=0$ and $\mu=0$
everywhere on $\Sigma$.
\end{proof}
\begin{cor}
Let $\Sigma\supset S^{1}$ be a contour invariant under inversion
in $S^{1}$. If $v\in H^{k}(\Sigma)$ with $\det v\ne0$ satisfies
\[
v=v^{\sharp}\;\text{on}\;\Sigma\quad\text{and}\quad v>0\;\text{on}\;S^{1},
\]
then $v=(m_{+})^{\sharp}m_{+}$ for some $m\in G\mathcal{H}^{k}(\mathbf{C}\setminus\Sigma)$.
\end{cor}

\begin{proof}
By the preceding theorem, $v$ has only zero partial indices and we
have $v=n_{-}n_{+}$ on $\Sigma$ for some $n=n(z)\in G\mathcal{H}^{k}(\mathbf{C}\setminus\Sigma)$.
Here we have replaced $n_{-}^{-1}$ by $n_{-}$. It is equivalent
to replacing $n$ by its inverse in $\Omega_{-}$. We have $v^{\sharp}=n_{+}^{\sharp}n_{-}^{\sharp}$.
Since $v=v^{\sharp}$, \cite[p.11]{clanceygohberg} implies that there
exists a constant matrix $C$ such that $n_{-}=n_{+}^{\sharp}C$ and
$n_{+}=C^{-1}n_{-}^{\sharp}$. Therefore we have $v=n_{+}^{\sharp}Cn_{+}$
on $\Sigma$. In particular, we have $v=n_{+}^{*}Cn_{+}$ on $S^{1}$.
Since $v$ is Hermitian and positive on $S^{1}$, so is $C$. There
exists a positive Hermitian matrix $R$ such that $R^{2}=C$. We have
$v=n_{+}^{\sharp}R^{2}n_{+}=(Rn)_{+}^{\sharp}(Rn)_{+}$ everywhere
on $\Sigma$. 
\end{proof}
\begin{example}
Let $\Sigma$ be the unit circle $S^{1}$, and set 
\[
v(z):=\begin{bmatrix}1-|r(z)|^{2} & -z^{2n}\bar{r}(z)\\
z^{-2n}r(z) & 1
\end{bmatrix}=\begin{bmatrix}1-|r'(z)|^{2} & -\bar{r'}(z)\\
r'(z) & 1
\end{bmatrix}\,(z\in S^{1}).
\]
Here $r(z)$ is a sufficiently smooth function on $S^{1}$, $r'(z)=z^{-2n}r(z)$
and $n$ is an integer. If $|r(z)|<1$, then $\mathrm{Re}\,v=\mathrm{diag}\,[1-|r(z)|^{2},1]>0$
and Theorem \ref{thm: inversion} applies. The matrix $v(z)$ is modeled
on the one corresponding to the defocusing integrable discrete nonlinear
Schr\"odinger equation (IDNLS). See \cite{APT,IDNLSdefocusing,IDNLSdefocusing2}.
But in the present paper, it is not necessary to assume that $r(z)$
is obtained by the scattering transform. It can be prescribed without
reference to a potential and we do not have to assume $r(-z)=-r(z)$
(\cite[(3.2.76)]{APT}), a property of the reflection coefficient.
The present author hopes this example and Theorem \ref{thm: IDNLSRHP}
below give a basis for establishing the bijectivity of the scattering/inverse
scattering transforms (cf. \cite{DZ2003,CPAM89,CPAM98}). 
\end{example}

\section{RHP modeled on the focusing IDNLS}

In this section, we consider a problem modeled on the focusing IDNLS
(\cite{APT}). Let $z_{j}\,(j=1,2,\dots,J)$ be distinct points outside
$S^{1}$. We consider the Riemann-Hilbert boundary value problem 
\begin{align}
 & M_{+}(z)=M_{-}(z)V(z)\;\text{on}\;S^{1},\label{eq: poleRHP1}\\
 & V(z)=\begin{bmatrix}1+|r(z)|^{2}\; & \;z^{2n}\bar{r}(z)\\
z^{-2n}r(z)\; & \;1
\end{bmatrix},\label{eq: poleRHP2}\\
 & \mathrm{Res}(M(z);z_{j})=\lim_{z\to z_{j}}M(z)\begin{bmatrix}0 & 0\\
z_{j}^{-2n}c_{j} & 0
\end{bmatrix},\label{eq: poleRHP3}\\
 & \mathrm{Res}(M(z);\bar{z}_{j}^{-1})=\lim_{z\to\bar{z}_{j}^{-1}}M(z)\begin{bmatrix}0 & \bar{z}_{j}^{-2n-2}\bar{c}_{j}\\
0 & 0
\end{bmatrix},\label{eq: poleRHP4}\\
 & M(z)\to I\;\mbox{\,\ as\,}\;z\to\infty.\label{eq: poleRHP5}
\end{align}
Here $n$ is an integer, $r(z)$ is a sufficiently smooth function
on the unit circle and $c_{j}$ is an arbitrary complex number. Moreover
$M_{+}$ and $M_{-}$ are the boundary values from the \textit{outside}
and \textit{inside} of the unit circle respectively. The unit circle
is oriented clockwise following the convention in \cite{APT}. In
the study of the focusing IDNLS, one encounters quartets of the form
$\left\{ \pm z_{j},\pm1/\bar{z}_{j}\right\} $, but in the present
paper we generalize the situation and consider pairs of the form $\left\{ z_{j},1/\bar{z}_{j}\right\} $.
Moreover we do not assume $r(-z)=-r(z)$. 

Following \cite{IDNLS}, we reduce this problem to one without poles. 

Let $C[z_{j}]$ be a sufficiently small circle centered at $z_{j}$
for each $j$. Assume that it is oriented clockwise. By inversion
in $S^{1}$, we get $C[z_{j}]^{\sharp}$, which is oriented \textsl{counter}clockwise.
This simple closed curve encloses $1/\bar{z}_{j}$. 

Set 
\[
m(z)=\begin{cases}
M(z)\begin{bmatrix}1 & 0\\
-\dfrac{z_{j}^{-2n}c_{j}}{z-z_{j}} & 1
\end{bmatrix} & \text{inside}\;C[z_{j}],\\[20pt]
M(z)\begin{bmatrix}1 & -\dfrac{\bar{z}_{j}^{-2n-2}\bar{c}_{j}}{z-\bar{z}_{j}^{-1}}\\
0 & 1
\end{bmatrix} & \text{inside}\;C[z_{j}]^{\sharp}.
\end{cases}
\]
and $m(z)=M(z)$ elsewhere. Then $m(z)$ is holomorphic near $z_{j},\bar{z}_{j}^{-1}$. 

Set $\Sigma=S^{1}\cup\cup_{j=1}^{J}C[z_{j}]\cup_{j=1}^{J}C[z_{j}]^{\sharp}$.
We introduce a matrix $v(z)$ on $\Sigma$ by 
\begin{align*}
v(z)= & \begin{cases}
V(z) & \text{on}\;S^{1},\\
\begin{bmatrix}1 & 0\\
\dfrac{z_{j}^{-2n}c_{j}}{z-z_{j}} & 1
\end{bmatrix} & \text{on}\;C[z_{j}],\\
\begin{bmatrix}1 & -\dfrac{\bar{z}_{j}^{-2n-2}\bar{c}_{j}}{z-\bar{z}_{j}^{-1}}\\
0 & 1
\end{bmatrix} & \text{on}\:C[z_{j}]^{\sharp}.
\end{cases}
\end{align*}
Then the RHP (\ref{eq: poleRHP1})-(\ref{eq: poleRHP5}) is equivalent
to the following RHP without poles:
\begin{equation}
m_{+}(z)=m_{-}(z)v(z)\;\text{on}\;\Sigma\quad\text{and \; }m(z)\to I\,(z\to\infty).\label{eq:IDNLSRHP}
\end{equation}

\begin{thm}
\label{thm: IDNLSRHP}The classical Riemann-Hilbert problem (\ref{eq:IDNLSRHP})
has a unique solution and so does (\ref{eq: poleRHP1})-(\ref{eq: poleRHP5}).
Moreover $v$ has only zero partial indices. 
\end{thm}

\begin{proof}
The jump matrix $v(z)$ does not satisfy the assumption of Theorem
\ref{thm: inversion}, but can be converted to such a one by conjugation,
i.e. by introducing a new unknown matrix $m'(z)$. 

Let $C_{R}$ and $C_{1/R}$ be the circles $|z|=R$ and $|z|=1/R$
respectively, where $R>0$ is sufficiently large. We give them both
counterclockwise orientation. 

\includegraphics[scale=0.35]{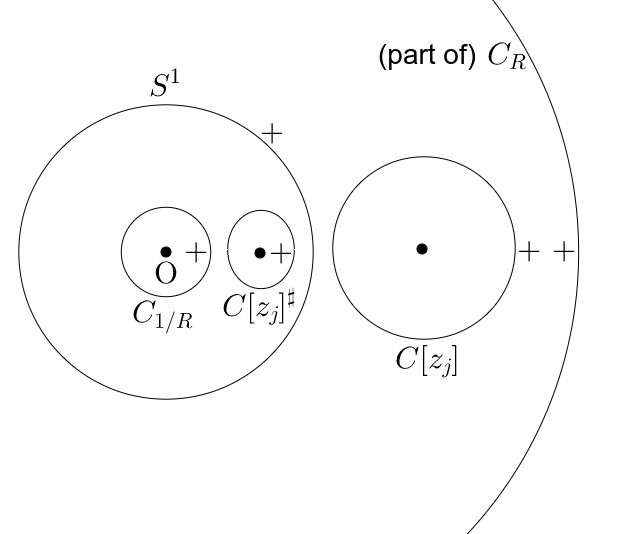}

We introduce 
\begin{align*}
A=A(z)= & \begin{bmatrix}\prod_{k=1}^{J}z_{k} & 0\\
0 & z
\end{bmatrix},\\
B_{j}=B_{j}(z)= & \begin{bmatrix}\prod_{k=1}^{J}z_{k} & 0\\
-\left(\prod_{k\ne j}z_{k}\right)z_{j}^{-2n}c_{j} & z
\end{bmatrix}\,(j=1,\dots J),\\
C=C(z)= & \begin{bmatrix}\prod_{k=1}^{J}\bar{z}_{k}^{-1} & 0\\
0 & z
\end{bmatrix}.
\end{align*}
We define $m'=m'(z)$ by the following set of rules: (1) $m'=m$ inside
$C_{1/R}$ and outside $C_{R}$. (2) $m'=mA$ if $z$ is between $S^{1}$
and $C_{R}$ and is outside $C[z_{j}]$ for all $j$. (3) $m'=mB_{j}$
inside $C[z_{j}]$. (4) $m'=mC$ between $C_{1/R}$ and $S^{1}$ except
on  $\cup_{j}C[z_{j}]^{\sharp}$. Then the normalization condition
at $\infty$ remains the same. Now we calculate the jump matrix $v'=v'(z)$
for $m'$: $m'_{+}=m'_{-}v'$ on $\Sigma$. 

We have $v'=A$ on $C_{R}$ and $v'=C^{-1}$ on $C_{1/R}$. Since
 $A^{\sharp}=C^{-1},$ we have $v'^{\sharp}(z)=v(z)$ for $z\in C_{R}\cup C_{1/R}$. 

On $C[z_{j}]$, we have $v'=B_{j}^{-1}vA$. We evaluate $v'^{-1}=A^{-1}v^{-1}B_{j}$
first, because $A^{-1}$ is easier than $B_{j}^{-1}$. Then we get
\[
v'=(v'^{-1})^{-1}=\begin{bmatrix}1 & 0\\
\dfrac{\left(\prod_{k\ne j}z_{k}\right)z_{j}^{-2n}c_{j}}{z-z_{j}} & 1
\end{bmatrix}
\]
on $C[z_{j}]$. Next on $C[z_{j}]^{\sharp}$, we have 
\[
v'=C^{-1}vC=\begin{bmatrix}1 & -\dfrac{z\left(\prod_{k\ne j}\bar{z}_{k}\right)\bar{z}_{j}^{-2n-1}\bar{c}_{j}}{z-\bar{z}_{j}^{-1}}\\
0 & 1
\end{bmatrix}.
\]
Therefore $v'^{\sharp}(z)=v'(z)$ holds for $z\in C[z_{j}]\cup C[z_{j}]^{\sharp}$. 

On $S^{1},$we have $v'=C^{-1}vA=A^{\sharp}vA=A^{*}vA$. Since $v$
is a positive Hermitian matrix, so are $v'$ and $\mathrm{Re}\,v'$. 

By Theorem \ref{thm: inversion}, the matrix $v'$ has only zero partial
indices. By Proposition \ref{prop: RHPclassical} and Corollary \ref{cor: RHPuniquesol},
the classical RHP $m'_{+}=m'_{-}v',m'\to I(z\to\infty)$ has a unique
solution and so does (\ref{eq:IDNLSRHP}). 

We have $v'=PvQ$, where $\left\{ P,Q\right\} \subset\left\{ A,B_{j}^{-1},C,C^{-1}\right\} $.
Let $v'=\tilde{m}_{-}^{-1}\tilde{m}_{+}$ be its factorization. We
have $v=P^{-1}\tilde{m}_{-}^{-1}\tilde{m}_{+}Q^{-1}$. This factorization
of $v$ means all the partial indices are zero. 
\end{proof}
\begin{rem}
In Theorem \ref{thm: IDNLSRHP} above, $\left\{ (z_{j},1/\bar{z}_{j}),c_{j};\,j=1,\dots,J\right\} $
and $r(z)$ are not true scattering data. True ones have two additional
characteristics: poles appear in quartets of the form $(\pm z_{j},\pm1/\bar{z}_{j})$
and the reflection coefficent satisfy $r(-z)=-r(z)$. According to
Theorem \ref{thm: IDNLSRHP}, we can solve the associated RHP uniquely
even for this kind of formal or generalized `scattering data' and
apply the potential reconstruction formula, but the `potential' obtained
this way is not necessarily a potential. 
\end{rem}

\end{document}